\documentclass[12pt]{article}
\usepackage{mathrsfs}
\usepackage{dsfont}
\usepackage{amsthm}
\usepackage{mathrsfs}
\usepackage{amsmath}
\usepackage{amsfonts}
\usepackage[colorlinks,linkcolor=black,anchorcolor=blue,citecolor=blue]{hyperref}
\usepackage{amssymb, amsmath, cite}

\newtheorem{theorem}{Theorem}[section]
\newtheorem{lemma}{Lemma}[section]

\newcommand\be{\begin{equation}}
\newcommand\ee{\end{equation}}
\newcommand\ber{\begin{eqnarray}}
\newcommand\eer{\end{eqnarray}}
\newcommand\berr{\begin{eqnarray*}}
\newcommand\eerr{\end{eqnarray*}}
\newcommand\Om{\Omega}\newcommand\lm{\lambda}
\newcommand\vp{\varphi}
\newcommand\vep{\varepsilon}
\newcommand{\dd}{\mathrm{d}}\newcommand\bfR{\mathbb{R}}
\newcommand\e{\mathrm{e}}\newcommand\pa{\partial}
{\setlength\arraycolsep{2pt}

\begin{document}

\title{Determination of AdS Monopole Wall via Minimization}
\author{Lei Cao\\Institute of Contemporary Mathematics\\Henan University\\
Kaifeng, Henan 475004, PR China\\ \\Yisong Yang\\
 Courant Institute of Mathematical Sciences\\ New York University\\
New York, NY 10012, USA}
\date{}
\maketitle

\begin{abstract}
In this note we solve a minimization problem arising in a recent work of Bolognesi and Tong on the determination of an AdS monopole wall. We show that the problem
has a unique solution. Although the solution cannot be obtained explicitly, we show that it may practically be constructed via a shooting method for which the correct shooting
slope is unique. We also obtain some  energy estimates which allow an asymptotic  explicit  determination of the monopole wall in a large coupling parameter limit.
\end{abstract}

\medskip
\begin{enumerate}

\item[]
{Keywords:} Gauge field theory, AdS metric, monopole wall, minimization problem, uniqueness of minimizer, asymptotic energy estimates.

\item[]
PACS numbers: 02.30.Hq, 11.15.−q, 11.27.+d, 12.39.Ba.

\item[]
{MSC numbers:} 34B40, 35J50, 81T13.

\end{enumerate}

\section{The minimization problem}\label{s1}

In this work, we are concerned with a minimization problem arising in the determination of the location of a
monopole wall
in the context of the 't Hooft--Polyakov monopole model \cite{tH,P},
described by an $SU(2)$ gauge field $A^{a}_\mu$ and an adjointly represented Higgs field $\phi$, over an Anti-de Sitter (AdS) spacetime given by the metric, following Bolognesi and Tong \cite{BT},
\be\label{1.1}
\dd s^2=-\left(1+\frac{\rho^2}{L^2}\right)\dd \tau^2+\left(1+\frac{\rho^2}{L^2}\right)^{-1}\dd\rho^2 +\rho^2\dd\Om_2^2,
\ee
where $\rho$ is the radial variable of $\bfR^3$,  $\dd\Om_2^2$ denotes the usual area element of $S^2$
expressed in terms of  polar and azimuthal coordinates, and $L>0$
is a scaling parameter, of dimension of length and often referred to as the radius of curvature. Now, reparametrizing $S^2$ with
a pair of planar polar angle $\theta$ and radial variable $\chi$ ($0\leq \chi<1$), we have
\be\label{1.2}
\dd\Om^2_2=\frac{\dd \chi^2}{1-\chi^2}+\chi^2\dd\theta^2.
\ee
Substituting (\ref{1.2}) into (\ref{1.1}), introducing the new variables $r, t, u$ such that
\be\label{1.3}
\rho=\zeta r,\quad \tau=\frac t\zeta,\quad \chi=\frac u{\zeta L},
\ee
with $\zeta>0$ a scaling factor, and setting $\zeta\to\infty$, we are led to the limiting metric
\ber\label{1.4}
\dd s^2=\frac{r^2}{L^2}(-\dd t^2+\dd u^2+u^2\dd \theta^2)+\frac{L^2}{r^2}\dd r^2.
\eer
It is demonstrated in \cite{BT} that, in the large monopole charge limit where the monopole number is of the magnitude
$\zeta^2$  and within Abelian approximation, the monopole wall sits at a certain location $r=R>0$ that divides the space into two 
regions, namely,  the infra-red region $r<R$ characterized by vanishing Higgs and gauge fields, and the ultra-violet region $r>R$ where the gauge field generates a constant magnetic field $B$ and the Higgs field $\phi$ is dynamically governed jointly by the
coupling parameters and potential energy. See also \cite{Bol} for an earlier formalism arguing for the existence of monopole walls. Specifically, the Yang--Mills--Higgs action density of the 't Hooft--Polyakov
monopole model \cite{tH,P} assumes the form
\ber\label{1.5}
{\cal L}=-\frac{1}{4}F_{\mu\nu}^aF^{a\mu\nu}-\frac{1}{2}{D}_\mu \phi^a{D}^\mu \phi^a-V(\phi),
\eer
 living over the correspondingly
reduced AdS spacetime,
where the Yang--Mills field, covariant derivative, and  potential density are given by
\be
F_{\mu\nu}^a=\pa_\mu A^a_\nu-\pa_\nu A^a_{\mu}+\epsilon^{abc}A_\mu^b A_\nu^c,\quad 
{D}_\mu\phi^a=\partial_\mu\phi^a+\epsilon^{abc}A_\mu^b\phi^c,\quad
V(\phi)=\frac{\lambda}{8}(|\phi|^2-v^2)^2,
\ee
respectively, with $\lm>0$ the Higgs coupling parameter and $v>0$ the Higgs vacuum of spontaneously broken symmetry. In such a context, 
with the assumed radial symmetry, the determination of the monopole wall position, $R$,
is shown \cite{BT} to be given by minimizing the energy
\ber\label{1.7}
E(R)=\frac{1}2\int_R^\infty\dd r\frac{r^2}{L^2}\left(\frac{L^4}{r^4}B^2+\frac{r^2}{L^2}\phi_r^2+\frac{\lambda}{4}(\phi^2-v^2)^2\right)+\frac\lm8\int_0^R\dd r\frac{r^2}{L^2}\,{ v^4},
\eer
in which the partition between the infra-red and ultra-violet regions are clearly exhibited. Here and in the sequel $\phi
=\phi(r)$ is taken to be a real-valued function representing the profile of the Higgs field which interpolates the vacua of symmetry and spontaneously broken
symmetry such that
\be\label{1.8}
\phi(R)=0,\quad \phi(\infty)=v,
\ee
which minimizes the part of the energy in (\ref{1.7}) containing $\phi$:
\be\label{1.9}
K(R)=\frac{1}2\int_R^\infty\dd r\left(\frac{r^4}{L^2}\phi_r^2+\frac{\lambda r^2}{4}(\phi^2-v^2)^2\right).
\ee
As in \cite{BT}, we may use the rescaled variables, $r=Rx$ and $\phi=v\varphi$, to recast (\ref{1.9}) into
\be\label{1.10}
K(R)=\frac{R^3 v^2}{2L^2}\int_1^\infty\dd x\left({x^4}\varphi_x^2+\frac{\beta x^2}{4}(\varphi^2-1)^2\right),
\ee
 where $\varphi=\varphi(x)$ satisfies the boundary condition $\varphi(1)=0,\varphi(\infty)=1$, updated from (\ref{1.8}), and
\be
\beta=\lm v^2 L^2.
\ee
is the rescaled Higgs coupling parameter. It is interesting that (\ref{1.9}) takes a factored form that singles out the dependence of
the energy functional (\ref{1.9}) on $R$ in a homogeneous manner. This useful property enables the authors of \cite{BT} to reexpress (\ref{1.7}) as
\ber\label{1.12}
E(R)=\frac{1}{2}\left(\frac{L^2B^2}{R}+\frac{R^3v^2}{L^4}f(\beta)+R^3\frac{\lambda v^4}{12L^2}\right),
\eer
where $f(\beta)$ is defined by
\ber\label{1.13}
f(\beta)=\min\left\{\int_1^\infty\,\dd x\left(x^4\varphi_x^2+\frac{\beta x^2}{4}(\varphi^2-1)^2\right)\,\bigg|\,\varphi\mbox{ satisfies }\varphi(1)=0,\varphi(\infty)=1\right\}.
\eer

Thus, with $f(\beta)$ defined in (\ref{1.13}), the position of monopole wall, $R>0$, is seen  to be obtained readily by minimizing
the elementary function $E(R)$ given in (\ref{1.12}) to yield the explicit formula \cite{BT}:
\be\label{1.14}
R=\sqrt{\frac{BL^3}v}\left(3f(\beta)+\frac\beta4\right)^{-\frac14}.
\ee

In this study, we shall establish the well-posedness of the function $f(\beta)$ by showing the existence and uniqueness of an energy
minimizer of the minimization problem defined by the right-hand side of (\ref{1.13}). We shall also obtain the sharp asymptotic estimate
\be\label{1.15}
f(\beta)=\mbox{O}(\sqrt{\beta})\quad \mbox{for $\beta$ large}.
\ee
Consequently, combining (\ref{1.14}) and (\ref{1.15}), we see that the monopole wall position satisfies the asymptotic estimate
\be\label{1.16}
R\approx\sqrt{\frac{2BL^3}v}\beta^{-\frac14},\quad \beta\gg 1.
\ee

In the next section, we prove the existence and uniqueness of a minimizer of the minimization problem given by the right-hand side
of (\ref{1.13}), and, in the subsequent section, we establish the estimate (\ref{1.15}).

\section{Existence and uniqueness of an energy minimizer}
\setcounter{equation}{0}

With the compressed notation $\vp'=\vp_x$, etc., the problem amounts to obtaining the existence and uniqueness of a minimizer of the energy functional
\ber\label{2.1}
I(\varphi)=\int_1^{\infty}\left(x^4(\varphi')^2+\frac{\beta x^2}{4}(\varphi^2-1)^2\right)\,\dd x,
\eer
over the admissible functions subject to the boundary condition 
\be\label{2.2}
\varphi(1)=0,\quad\varphi(\infty)=1.
\ee
For this problem, the corresponding Euler--Lagrange equation is the nonlinear equation
\ber\label{2.3}
(x^4\varphi')'=\frac{\beta x^2}{2}\varphi(\varphi^2-1),\quad 1<x<\infty.
\eer

We have the following.

\begin{theorem}\label{th2.1}
The two-point boundary value problem consisting of \eqref{2.2}--\eqref{2.3} over the interval $[1,\infty)$  has a unique solution
$\varphi$ which enjoys the following properties.

\begin{enumerate}
\item[(i)] $\varphi$ strictly increases  such that
$
0<\varphi(x)<1
$ for $x>1$.
\item[(ii)] $\varphi$ satisfies the sharp boundary estimates given respectively by
\be\label{A1}
\varphi(x)=a_0 (x-1)+\mbox{\rm O}((x-1)^{2})\quad \mbox{ as }x\to 1,\quad x>1,
\ee
where $a_0>0$ is a uniquely determined constant, and
\be\label{A2}
1-\varphi(x)=\mbox{\rm O}(x^{-(\lm_0-\vep)})\quad\mbox{ as }x\to\infty,\quad\mbox{with } \lm_0=\frac32+\sqrt{\left(\frac32\right)^2+\beta},
\ee
where $\vep>0$ is a constant which may be made arbitrarily small.
\item[(iii)] $\varphi$ minimizes the energy functional (\ref{2.1}) among the functions satisfying (\ref{2.2}).
\item[(iv)] The function $f(\beta)=I(\vp)$, where $\vp$ is as obtained above, defined in (\ref{1.13}) enjoys the bounds
\be
\frac{2\sqrt{\beta}}3<f(\beta)<\sigma(\beta)+\frac{2\sqrt{\beta}}3,\quad \beta>\frac94,
\ee
where $\sigma(\beta)$ satisfies 
\be\label{x2.7}
1+\frac3{2\sqrt{\beta}-3}<\sigma(\beta)<6+\frac{18}{2\sqrt{\beta}-3},\quad\beta>\frac94.
\ee
In particular, there holds the following sharp asymptotic estimate
\be
f(\beta)\sim \frac{2\sqrt{\beta}}3,\quad\beta\gg1.
\ee
\end{enumerate}
\end{theorem}

To prove this theorem, we consider the  initial value problem
\be\label{4.1}
(x^4\varphi')'=\frac{\beta x^2}{2}\varphi(\varphi^2-1),\quad\varphi(1)=0,\quad \varphi'(1)=a,\quad a>0,
\ee
where the constant $a$ is an initial slope. We show that there is a unique $a>0$ such that the solution of (\ref{4.1}) solves the boundary value problem \eqref{2.2}--\eqref{2.3} and \eqref{1.13}. In other words, our approach is
 a shooting method. As a by-product, this study also justifies the implementation of the shooting method in constructing the energy-minimizing solution numerically, which will be useful in practice.

For convenience, we use $\varphi(x; a)$ to denote the unique local solution of \eqref{4.1} and define the shooting sets as follows:
\ber
{\mathcal{A}}^{-}&=&\big\{a>0~|~\varphi'(x; a)<0~\text{for some}~x>1 \big\},\notag\\
{\mathcal{A}}^{0}&=&\big\{a>0~|~\varphi'(x; a)>0\mbox{ and } \varphi(x, a)\leq 1~\text{for all}~x>1 \big\},\notag\\
{\mathcal{A}}^{+}&=&\big\{a>0~|~\varphi'(x; a)>0~\text{for all}~x>1\mbox{ and } \varphi(x; a)>1~\text{for some}~x>1\big\},\notag
\eer
in which, and in the sequel, the statements are made to mean wherever the solution exists.

The proofs of (i)--(iii) of Theorem \ref{th2.1} are split into five lemmas.

\begin{lemma}\label{d.1}
There hold
${\mathcal{A}}^{+}\cap{\mathcal{A}}^{-}={\mathcal{A}}^{+}\cap{\mathcal{A}}^{0}
={\mathcal{A}}^{-}\cap{\mathcal{A}}^{0}=\varnothing$ and $(0,\infty)={\mathcal{A}}^{+}\cup{\mathcal{A}}^{0}\cup{\mathcal{A}}^{-}$.
\end{lemma}
\begin{proof}
It is clear that  
  $\mathcal{A}^{-}$, $\mathcal{A}^{0}$,  and $\mathcal{A}^{+}$ are nonoverlapping.
If $a>0$ but $a\notin\mathcal{A}^{-}$, then $\varphi'(x; a)\geq0$ for all $x>1$.  Suppose there is a point $x_0>1$ so that $\varphi'(x_0; a)=0$.
Then $\varphi(x_0, a)(\varphi^2(x_0, a)-1)\neq 0$, otherwise $\varphi$ arrives
at an equilibrium of the differential equation at $x_0$, which violates the uniqueness theorem for solutions to initial value problems of ordinary differential equations. 
So we have $\varphi''(x_0, a)>0$ or $\varphi''(x_0, a)<0$ in view of $\varphi'(x_0, a)=0$ and $\varphi(x_0, a)(\varphi^2(x_0, a)-1)\neq 0$ at $x=x_0$. Therefore,
 when $x$ is close to $x_0$, there holds $\varphi'(x; a)<0$ either for $x<x_0$ or $x>x_0$, contradicting the assumption $a\notin\mathcal{A}^{-}$. Thus $a\in \mathcal{A}^0\cup\mathcal{A}^+$. In
other words, there holds $(0, \infty)={\mathcal{A}}^{+}\cup{\mathcal{A}}^{0}\cup{\mathcal{A}}^{-}$.
\end{proof}

\begin{lemma}\label{d.2}
The sets $\mathcal{A}^{-}$ and $\mathcal{A}^{+}$ are both open and nonempty. Moreover, there are numbers $\varepsilon,\delta>0$ such that $(0,\varepsilon)\subset\mathcal{A}^-$ and
$(\delta,\infty)\subset\mathcal{A}^+$.
\end{lemma}
\begin{proof}
 In fact, integrating the differential equation in \eqref{4.1} gives us
\ber\label{4.2}
x^4\varphi'(x; a)=a+\frac\beta2\int_1^x { t^2}\varphi(t; a)(\varphi^2(t; a)-1)\,\dd t,\quad x>1.
\eer

We first show that $\mathcal{A}^-\neq\emptyset$. Assume otherwise $a\not\in\mathcal{A}^-$ for any $a>0$. Then $a\in\mathcal{A}^0\cup\mathcal{A}^+$. So $\varphi'(x, a)>0$ for all $x>1$.  Let $(1, K_a)$ denote the interval where
$0<\varphi(x;a)<1$ for $x\in (1,K_a)$. Then either $K_a\in (1,\infty)$ or  $K_a=\infty$. In the former case, $\varphi(K_a;a)=1$, and in the latter case, $\varphi(x;a)\to 1$ as $x\to\infty$. 
In both cases, we may find $x_1,x_2\in (0,K_a)$, $x_1<x_2$, such that
\ber\label{4.7}
\varphi(x_1; a)=\frac{1}{4};\quad \varphi(x_2; a)=\frac{1}{2};\quad \frac{1}{4}\leq \varphi(x; a)\leq\frac{1}{2},\quad  x\in[x_1, x_2].
\eer
On the other hand, in view of (\ref{4.2}), we have $0<\varphi'(x;a)<x^4\varphi'(x;a)<a$ for $x\in(1,K_a)$, which gives us after an integration over $[x_1,x_2]$  and using (\ref{4.7}) the lower bound
\ber\label{4.8}
x_2-x_1\geq\frac{1}{4a}.
\eer
Applying \eqref{4.7} and \eqref{4.8} in \eqref{4.2}, we obtain
\ber\label{4.9}
0<x_2^4\varphi'(x_2;a)&=&a+\frac\beta2\int_1^{x_2}{ t^2}\varphi(t; a)(\varphi(t; a)^2-1)\dd t \notag\\
&<&a+\frac{\beta}{2}\int_{x_1}^{x_2}\varphi(t; a)(\varphi(t; a)^2-1)\dd t\notag\\
&<&a-\frac{\beta}{2}\min\left\{|\varphi(\varphi^2-1)|\,\bigg|\, \frac{1}{4}\leq \varphi\leq\frac{1}{2}\right\}\int_{x_1}^{x_2}\dd t\notag\\
&<& a-\frac{15\beta}{2^9 a },
\eer
which cannot hold when $a>0$ is small enough. In other words, we have shown that $\mathcal{A}^{-}$ must contain an interval of the form $(0, \varepsilon)$ ($\varepsilon>0$). So $\mathcal{A}^{-}$ is nonempty
in particular. From the continuous dependence theorem of solutions of ordinary differential equations on initial values we see that the openness of $\mathcal{A}^{-}$ follows.

We next consider $\mathcal{A}^{+}$. For this purpose, first fix any $x_0>1$ and choose $a>0$ sufficiently large so that
\ber\label{4.11}
a>\frac{\beta}{2}\max\left\{|\varphi(\varphi^2-1)|\,\bigg|\,0\leq \varphi\leq 1\right\}\int_1^{x_0}t^2\,\dd t=\frac\beta{9\sqrt{3}}(x_0^3-1).
\eer
In view of (\ref{4.11}) and (\ref{4.2}), we obtain $\varphi'(x;a)>0$ for $x\in [1,x_0]$.
Furthermore,
integrating \eqref{4.2} and inserting the right-hand side of \eqref{4.11}, we get
\ber\label{4.10}
\varphi(x_0; a)&=&\int_1^{x_0}\frac1{x^4}\left({a}+\frac\beta2\int_1^x t^2\varphi(t; a)(\varphi^2(t; a)-1)\dd t\right)\dd x\notag\\
&\geq& \int_1^{x_0}\frac1{x^4}\left(a-\frac\beta{9\sqrt{3}}(x_0^3-1)\right)\,\dd x,
\eer
which may be made to exceed 1 when $a$ is large enough.  In view of this result and \eqref{4.2},  we see that $\varphi'(x; a)>0$ and $\varphi(x;a)>1$ for all $x>x_0$. In particular, $a\in{\mathcal{A}^+}$.
Thus we have shown that $(\delta, \infty)\subset\mathcal{A}^+$  when $\delta>0$ is sufficiently large.

To see $\mathcal{A}^+$ is open, let $a_0\in\mathcal{A}^+$ and $\varphi(x; a_0)$ be the corresponding solution of \eqref{4.1} with $a=a_0$. For $\sigma>0$ sufficiently small, we can find a unique $x_1>0$ so that $\varphi(x_1; a_0)=1+\sigma$. By the continuous dependence of solutions of ordinary differential equations on initial values, we can find a small open neighborhood $U$ of $a_0$ so that  $\varphi'(x; a)>0$ for $x\in[1,x_1]$ and $\varphi(x_1; a)>1+\frac{\sigma}{2}$.
From \eqref{4.2}, we see that $\varphi'(x; a)>0$ and $\varphi(x;a)>1+\frac\sigma2$ continue to hold for all $x>x_1$. This proves $U\subset\mathcal{A}^+$ and thus the openness of $\mathcal{A}^+$ follows as well.
\end{proof}

It will be useful to get some estimate for $a\in\mathcal{A}^+$. In fact, for $a\in\mathcal{A}^+$, let $x_1(a)>1$ be the unique point such that 
\be
\varphi(x_1(a);a)=1.
\ee
Then (\ref{4.2}) and the property $0<\varphi(x;a)<1$ for $x\in(1,x_1(a))$ lead to
\ber
1&=&\int^{x_1(a)}_1 \frac1{x^4}\left({a}+\frac\beta2\int_1^x t^2\varphi(t; a)(\varphi^2(t; a)-1)\dd t\right)\dd x\notag\\
&<& \frac a3\left(1-\frac1{x^3_1(a)}\right).
\eer
Thus, we obtain
\be
a>3,\quad a\in{\cal A}^+,
\ee
and 
\be\label{2.15}
x_1(a)>\left(\frac{a}{a-3}\right)^{\frac13},\quad a\in{\cal A}^+.
\ee


With this last estimate and using the method in \cite{Chen,CHMY}, we are ready to prove the following.

\begin{lemma}\label{d.3}
The sets $\mathcal{A}^-$ and $\mathcal{A}^+$ are actually open intervals. More precisely, there are numbers $0<a_1\leq a_2<\infty$ such that $\mathcal{A}^-=(0,a_1)$ and $\mathcal{A}^+=(a_2,\infty)$.
\end{lemma}

\begin{proof}
Let $a\in{\cal A}^+$.  Using the notation
\be\label{2.16}
\vp_a\equiv\frac{\pa\vp(x;a)}{\pa a},
\ee
we see that (\ref{4.1}) gives us
\be\label{2.17}
(x^4 \vp'_a)'-\frac\beta2 x^2(3\vp^2-1)\vp_a=0,\quad \vp_a=0\mbox{ and } \vp_a'=1 \mbox{ at } x=1.
\ee
On the other hand, $\psi\equiv \frac\vp a$ satisfies 
\be\label{2.18}
(x^4 \psi')'-\frac\beta2 x^2(3\vp^2-1)\psi=-\frac\beta a x^2 \vp^3<0,\quad \psi=0\mbox{ and } \psi'=1 \mbox{ at } x=1.
\ee
Then a comparison argument applied to (\ref{2.17}) and (\ref{2.18}) gives \cite{Chen}
\be\label{2.19}
\vp_a(x;a)>\psi(x)=\frac{\vp(x;a)}a>0,\quad x\in (1,x_1(a)).
\ee
Furthermore, differentiating the equation
\be\label{2.20}
\vp(x_1(a);a)\equiv 1,\quad a\in{\cal A}^+,
\ee
with respect to $a$, we obtain 
$
\vp'(x_1(a);a)x'_1(a)+\vp_a(x_1(a);a)=0,
$
resulting in $x'_1(a)<0$ by \eqref{2.19}. In particular,  $x_1(a)$ decreases. Now we show that $\cal{A}^+$ is connected. For this purpose, it suffices to show that, if  $\cal{A}^+$ contains an interval of the form $(b_1,b_2)$ with $0<b_1<b_2<\infty$,
then $b_2\in {\cal A}^+$. In fact, for $a\in (b_1,b_2)$, we see that $x_1(a)$ decreases and satisfies the lower bound (\ref{2.15}). Using $a>b_1$,  we see that 
\be\label{2.21}
x_0\equiv \lim_{a\to b_2^-} x_1(a)
\ee
exists and lies in $(1,\infty)$. Inserting (\ref{2.21}) into (\ref{2.20}), we find $\vp(x_0;b_2)=1$. Besides, by continuity, we also have $\vp'(x;b_2)\geq0$ for $x>1$. Hence $x_0$ is the unique point for $x\in (1,\infty)$ where $\vp(x;b_2)=1$.
Of course $\vp'(x_0,b_2)\neq0$ otherwise $\vp\equiv1$ which is false. Hence $\vp'(x_0;b_2)>0$. From (\ref{4.1}), we have $(x^4\vp')'<0$ for $x\in(1,x_0)$. So $ x^4\vp'(x;b_2)>x_0^4\vp'(x_0;b_2)>0$ for $x\in(1,x_0)$. That is, $\vp'(x;b_2)>0$ for
$x\in(1,x_0)$. If $x>x_0$, then $\vp>1$. Thus (\ref{4.1}) gives us again $x^4\vp'(x;b_2)>x_0^4\vp'(x_0;b_2)>0$ for $x>x_0$. In other words, $\vp'(x;b_2)>0$ for all $x>1$ and $\vp(x;b_2)>1$ whenever $x>x_0$. This proves $b_2\in{\cal A}^+$ as desired.

We then show that $\cal{A}^-$ is also connected. For this purpose, note that for $a\in {\cal A}^-$, there is a point $y_1>1$ such that $\vp'(y_1;a)=0$. Use $y_1(a)$ to denote the left-most such a point. It is clear that
$0<\vp(x;a)<1$ for $x\in(1,y_1(a))$. In fact, since $1$ is an equilibrium of the equation in (\ref{4.1}), we have $\vp(y_1(a);a)\neq1$ because we have $\vp'(y_1(a);a)=0$ already. If $\vp(y_1(a);a)>1$, then there is a point $x_1\in (1,y_1(a))$ such
that $\vp(x;a)>1$ when $x\in(x_1,y_1(a))$. From the differential equation in (\ref{4.1}), we get $(x^4\vp')'>0$ for $x\in (x_1,y_1(a))$. So $x^4\vp'(x;a)<y_1^4(a)\vp'(y_1(a);a)=0$
for $x\in(x_1,y_1(a))$, resulting in the existence of a point $y_2\in (1,x_1)$ such that
$\vp'(y_2;a)=0$, which is false. Thus $\vp(y_1(a);a)<1$. Note also that a similar argument as before shows that $y_1(a)$ increases with respect to $a\in{\cal A}^-$. Now let $(b_1,b_2)\subset{\cal A}^-$. By Lemma \ref{d.2}, we may assume
\be\label{2.22}
b_1\geq \vep_0
\ee
for some $\vep_0>0$. It suffices to show $b_1\in{\cal A}^-$. In fact, taking $a\in(b_1,b_2)$ and integrating the differential equation in (\ref{4.1}) over $(1,y_1(a))$ and using $0<\vp(x;a)<1$ for $x\in(1,y_1(a))$, we have
\ber\label{2.23}
a&=&\frac\beta2\int_1^{y_1(a)}x^2\vp(x;a)(1-\vp^2(x;a))\,\dd x\notag\\
&<&\frac\beta2\max\{\vp(1-\vp^2)\,|\,0\leq \vp\leq1\}\int_1^{y_1(a)} x^2\,\dd x\notag\\
&=&\frac\beta{9\sqrt{3}}(y_1^3(a)-1).
\eer
In view of (\ref{2.22}) and (\ref{2.23}), we obtain the lower bound
\be\label{2.24}
y_1(a)>\left(1+\frac{9\sqrt{3}}\beta\vep_0\right)^{\frac13}.
\ee
Another by-product of the property $0<\vp(x;a)<1$ for $x\in(1,y_1(a))$ is that $\vp'(x;a)>0$ for $x\in(1,y_1(a))$. Now set
\be
y_0\equiv \lim_{a\to b_1^+} y_1(a).
\ee
Then (\ref{2.24}) gives us $y_0\in (1,\infty)$. Therefore, using continuity, we get $\vp'(y_0;b_1)=0$. Besides, for any $x\in(1,y_0)$, we have $x\in (1,y_1(a))$ when $a$ is
close to $b_1$. Thus
$0<\vp(x;a)<1$ and $\vp'(x;a)>0$ there. Letting $a\to b_1^+$, we find $0\leq\vp(x;b_1)\leq1$ and $\vp'(x;b_1)\geq0$ for $x\in(1,y_0)$.
Since $0$ and $1$ are equilibria of the differential equation in (\ref{4.1}) and $\vp'(y_0,b_1)=0$, we must have $0<\vp(y_0;b_1)<1$, which in turn implies $0<\vp(x;b_1)<1$ for $x\in(1,y_0)$.
Applying this result in the differential equation in (\ref{4.1}) again, we get $x^4 \vp'(x;b_1)>y_0^4\vp'(y_0;b_1)=0$ for $x\in(1,y_0)$ and $\vp''(y_0;b_1)<0$. Hence $\vp'(x;b_1)<0$ for $x\in (y_0,y_0+\delta)$
where $\delta>0$ is small enough. This establishes $b_1\in{\cal A}^-$.
\end{proof}

\begin{lemma}\label{d.4}
The set $\mathcal{A}^{0}$ has exactly one point. In other words, with the statement in Lemma \ref{d.3} such that ${\cal A}^- =(0,a_1)$ and ${\cal A}^+ =(a_2,\infty)$, we have
$a_1=a_2\equiv a_0$ so that ${\cal A}^0=\{a_0\}$. Moreover, with $a=a_0$ in (\ref{4.1}), we have $\vp(x;a_0)\to1$ as $x\to\infty$.
\end{lemma}
\begin{proof}
Let $a\in\mathcal{A}^0$. The definition of $\mathcal{A}^0$ indicates that there is some $\varphi_{\infty}\in(0,1]$ such that $\vp(x;a)\to\vp_\infty$ as $x\to\infty$. If $\varphi_{\infty}<1$, then 
the integral on the right-hand side of  \eqref{4.2} diverges as $x\to\infty$. In particular, $\varphi'(x;a)<0$ when $x$ sufficiently large, which is false.  Hence we have
\be\label{2.26}
\lim_{x\to\infty}\vp(x;a)=1,\quad a\in{\cal A}^0.
\ee
On the other hand, in view of Lemma \ref{d.3},  there exist $a_1$ and $a_2$ so that $0<a_1\leq a_2<\infty$ and $\mathcal{A}^{0}=[a_1, a_2]$.  We now show $a_1=a_2$.  In fact, assume otherwise $a_1<a_2$. 
Then with the notation (\ref{2.16}), we see that (\ref{2.19}) is valid for $x_1(a)=\infty$. Thus, for any $x>1$, we have
\ber\label{2.27}
\varphi(x; a_2)&=&\varphi(x; a_1)+\int_{a_1}^{a_2}\varphi_a(x; a)\dd a\notag\\
&\geq&\varphi(x; a_1)+\int_{a_1}^{a_2}\frac{1}{a}\varphi(x; a)\dd a,\quad x>1.
\eer
Letting $x\to\infty$ in \eqref{2.27} and applying (\ref{2.26}), we obtain
\ber
1\geq1+\ln\frac{a_2}{a_1},
\eer
which is false. Thus ${\cal A}^0$ can only be a set of a single point.
\end{proof}

We now consider asymptotic behavior of the unique solution $\vp$ to the two-point boundary value problem (\ref{2.2})--(\ref{2.3}). 

First, in view of the shooting method solution based on the initial-value problem (\ref{4.1}) obtained in Lemmas \ref{d.1}--\ref{d.4} and the Cauchy--Kovalevskaya theorem, we see
that the asymptotic expansion (\ref{A1}) holds. Next, with $x=\e^t$, we may rewrite (\ref{2.3}) as
\be\label{2.30}
\frac{\dd^2\vp}{\dd t^2}+3\frac{\dd\vp}{\dd t}=\frac\beta2\vp(\vp+1)(\vp-1),\quad t=\ln x,
\ee
which indicates that,
for $x\gg1$ or $t\gg1$, the equation \eqref{2.30} in $\psi=\vp-1$ assumes the asymptotic form
\be
\frac{\dd^2 \psi}{\dd t^2}+3\frac{\dd \psi}{\dd t}-\beta\psi=0,\quad t=\ln x,
\ee
for which the solution vanishing at infinity is given as
\be\label{2.32}
\psi (t)=C\e^{-\lm_0 t},\quad \lm_0=\frac32+\sqrt{\left(\frac32\right)^2+\beta}, 
\ee
where $C$ is a constant. Consequently, we may consider a comparison function of the form
\be
\psi_\vep(t)=C\e^{-(\lm_0-\vep)t},\quad \vep\in(0,\lm_0).
\ee
Then $\psi_\vep$ satisfies the equation
\be\label{2.34}
\frac{\dd^2\psi_\vep}{\dd t^2}+3\frac{\dd\psi_\vep}{\dd t}=c_\vep\psi_\vep,\quad c_\vep=\left(\beta+\vep[3-2\lm_0+\vep]\right).
\ee
Let $t_\vep>0$ be sufficiently large and $\vep>0$ small so that
\be\label{2.35}
\frac\beta2 \vp(t)(\vp(t)+1)>c_\vep,\quad t>t_\vep;\quad c_\vep >0.
\ee
Now set $u=1-\vp-\psi_\vep$. Then we are led from (\ref{2.30}), (\ref{2.34}), and (\ref{2.35}) to the inequality
\be\label{2.36}
\frac{\dd^2 u}{\dd t^2}+3\frac{\dd u}{\dd t}>c_\vep u,\quad t>t_\vep.
\ee
Choose $C>0$ in (\ref{2.32})  large enough such that $u(t_\vep)=1-\vp(t_\vep)-\psi_\vep(t_\vep)<0$. With this and the boundary property $u(t)\to0$ as $t\to\infty$ and applying the maximum principle to (\ref{2.36}), we have
$u(t)\leq 0$ for all $t>t_\vep$. In other words,  we arrive at
\be\label{2.37}
1-\vp(t)\leq C(\vep)\e^{-(\lm_0-\vep)t},\quad t>t_\vep.
\ee
Returning to the variable $x=\ln t$, we get the estimate (\ref{A2}).

Thus, (i) and (ii) in Theorem \ref{th2.1} are established. We now prove (iii). Since the energy functional (\ref{2.1}) is not quadratic, it does not allow a direct proof that our solution obtained in (i) and (ii) indeed minimizes the energy. Below, we 
pursue an indirect proof of this fact.

\begin{lemma}\label{d.5}
The solution to the boundary value problem (\ref{2.2})--(\ref{2.3}) obtained in (i) and (ii) in Theorem \ref{th2.1} minimizes the energy (\ref{2.1}) among the
functions satisfying (\ref{2.2}).
\end{lemma}
\begin{proof}
 Let $\{\varphi_n\}$ be a minimizing sequence of the problem
\be\label{2.38}
\eta\equiv\inf\{I(\vp)\,|\,\vp(1)=0,\vp(\infty)=1\}.
\ee
 Assume for all $n$ we have $E(\varphi_n)\leq \eta+1$. Then the Schwarz inequality leads to the uniform bound
\be\label{2.39}
|\varphi_n(x)-1|=\left|\int_x^\infty\varphi'_n(y)\dd y\right|
\leq\left(\int_x^\infty\frac{1}{y^4}\dd y\right)^{\frac{1}{2}}\left(\int_x^\infty y^4(\varphi_n'(y))^2\dd y\right)^{\frac{1}{2}}
\leq\left(\frac{\eta+1}{3x^3}\right)^{\frac12},
\ee
for $x>1$.
Similarly, we have
\be\label{2.40}
|\vp_n(x)|=\left|\int_1^x \vp_n'(y)\,\dd y\right|\leq\left( [\eta+1][x-1]\right)^{\frac12},\quad x>1.
\ee

For any $K>1$, we see that $\{\varphi_n\}$ is bounded in the Sobolev space $W^{1, 2}(1, K)$. Using a diagonal subsequence argument, we may find a subsequence of $\{\varphi_n\}$, which may still be denoted as $\{\varphi_n\}$, and an element
$\vp\in W^{1,2}_{\mbox{loc}}(1,\infty)$, such that  $\varphi_n \to \vp$ weakly  in $W^{1, 2}(1, K)$ for any $K>1$ as $n\to\infty$. By
the compact embedding of $W^{1, 2}(1, K)$ into $C[1, K]$, we have $\varphi_n\rightarrow\varphi$ in $C[1, K]$ as $n\to\infty$. Using these and Fatou's lemma in
\be
I(\varphi_n)\geq\int_1^{K}\left(x^4(\varphi'_n)^2+\frac{\beta x^2}{4}(\varphi_n^2-1)^2\right)\,\dd x,
\ee
we have
\be
\eta\geq\int_1^{K}\left(x^4(\varphi')^2+\frac{\beta x^2}{4}(\varphi^2-1)^2\right)\,\dd x.
\ee
Letting $K\to\infty$, we get $\eta\geq I(\vp)$. However, by the uniform estimates (\ref{2.39}) and (\ref{2.40}), we have $\vp(1)=0$ and $\vp(\infty)=1$. Therefore $\vp$ belongs to the admissible space of  the problem \eqref{2.38}.
Thus $I(\vp)=\eta$. That is, $\vp$ solves (\ref{2.38}).

Let $\vp$ solve \eqref{2.38}. Then the structure of the energy (\ref{2.1}) indicates that the replacement of $\vp$ by $|\vp|$ does not increase the energy. This implies that we may assume $\vp(x)\geq0$ for all $x\geq1$. 
Using $\vp(1)=0$, we get $\vp'(1)\geq0$. On the other hand, as a critical
point of the energy \eqref{2.1}, $\vp$ satisfies \eqref{2.3}. Thus $\vp'(1)>0$ otherwise $\vp\equiv0$ because 0 is an equilibrium of \eqref{2.3}. In other words, $\vp$ satisfies (\ref{4.1}) for some $a>0$. So $\vp$ is the unique
solution obtained in (i) and (ii) of Theorem \ref{th2.1} as claimed in (iii).
\end{proof}

The proof of (iv) of Theorem \ref{th2.1} will be presented in the next section.

\section{Energy estimates}
\setcounter{equation}{0}

We now estimate the minimum energy (\ref{1.13}) with $\vp$ be the solution of (\ref{2.2})--(\ref{2.3}) obtained in (i)--(iii) of Theorem \ref{th2.1}. For this purpose, we first apply the Bogomol'nyi \cite{Bo} trick to get for $I(\vp)$ defined in (\ref{2.1}) to be
\ber\label{5.1}
I(\varphi)
&=&\int_1^{\infty}\Big(\big(x^2\varphi'+\frac{\sqrt{\beta}x}{2}(\varphi^2-1)\big)^2-\sqrt{\beta}x^3(\varphi^2-1)\varphi'\Big)\dd x\notag\\
&=&\int_1^\infty\big(x^2\varphi'+\frac{\sqrt{\beta}x}{2}(\varphi^2-1)\big)^2\dd x+\sqrt{\beta}\int_1^\infty x^2(\varphi^3-3\varphi+2)\dd x+\frac{2\sqrt{\beta}}{3}\notag\\
&\equiv &I_1(\varphi)+I_2(\varphi)+\frac{2\sqrt{\beta}}{3},
\eer
where we have used the asymptotic estimate (\ref{A2}) and the factorization $\varphi^3-3\varphi+2=(\vp-1)^2(\vp+1)$ to get
\be\label{x3.2}
\lim_{x\to\infty} x^3 (\varphi^3(x)-3\varphi(x)+2)=0,
\ee
for a boundary term arising in \eqref{5.1}. Since both $I_1(\vp)$ and $I_2(\vp)$ are positive, we derive from (\ref{5.1}) the lower bound
\be\label{5.4}
f(\beta)=I(\vp)>\frac{2\sqrt{\beta}}{3}\quad\mbox{for any } \beta>0.
\ee

We now derive some upper estimates for $f(\beta)$. For this purpose,  let $\vp(x)$ be a function satisfying the boundary condition \eqref{2.2} and stays in $[0,1]$. Then the reduction \eqref{5.1} still holds 
under the condition (\ref{x3.2}) which we will observe. In order to find a sharper upper estimate of the minimum energy, we impose $I_1(\vp)=0$ which leads to the first-order equation
\ber
x^2\varphi'+\frac{\sqrt{\beta}x}{2}(\varphi^2-1)=0,\quad x>1,
\eer
whose solution subject to (\ref{2.2}) is
\ber\label{x5.6}
\varphi(x)=\frac{x^{\sqrt{\beta}}-1}{x^{\sqrt{\beta}}+1},\quad x\in[0,\infty),
\eer
such that \eqref{x3.2} is ensured by the condition
\be\label{3.6}
\beta>\frac94.
\ee

Assuming \eqref{3.6}, inserting \eqref{x5.6} into $I_2(\vp)$, and using the new variable $y=x^{\sqrt{\beta}}$, we obtain
\be\label{5.7}
I_2(\varphi)
=4\int_1^\infty y^{\frac{3}{\sqrt{\beta}}-1}\frac{3y+1}{(y+1)^3}\dd y\equiv \sigma(\beta).
\ee
This integral cannot be evaluated exactly. However, noting $\frac12<\frac{y^2(3y+1)}{(y+1)^3}<3$ for $y\in(1,\infty)$ in (\ref{5.7}), we have the estimate
\be\label{x}
\frac{2\sqrt{\beta}}{2\sqrt{\beta}-3}<\sigma(\beta)<\frac{12\sqrt{\beta}}{2\sqrt{\beta}-3},\quad \beta>\frac94.
\ee
Using this result in (\ref{5.1}), we arrive at the upper bound
\be\label{3.9}
f(\beta)< I(\vp)=\sigma(\beta)+\frac{2\sqrt{\beta}}3\quad\mbox{ for } \beta>\frac94,
\ee
where $\sigma(\beta)$ satisfies \eqref{x}, and here we have $f(\beta)<I(\vp)$ because now $\vp$ given in (\ref{x5.6}) is not a solution to (\ref{2.3}). 

We may also rewrite \eqref{x} more explicitly as (\ref{x2.7}). Thus (iv) of Theorem \ref{th2.1} follows.

By taking other concrete trial configurations, some other upper estimates for $f(\beta)$ may be obtained without assuming (\ref{3.6}). For example, we may choose 
\be
\vp=1-x^{-\alpha},\quad x\in [1,\infty),\quad \alpha>\frac32,
\ee
where the restriction to $\alpha$ is to ensure $I(\vp)<\infty$,  or
\be
\vp(x)=\frac{x-1}\alpha,\quad 1\leq x\leq \alpha+1;\quad \vp(x)=1,\quad x>\alpha+1,
\ee
where $\alpha>0$ is arbitrary. With the former choice of the trial function, we obtain the estimate
\be\label{3.12}
f(\beta)<I(\vp)=\frac\beta{4(4\alpha-3)}+\frac{\alpha^2+\beta}{2\alpha-3}-\frac\beta{3(\alpha-1)}\quad\mbox{for any }\beta>0.
\ee
We may further minimize the right-hand side of \eqref{3.12} to get a sharper upper estimate for $f(\beta)$ which we omit here. Note that, although this estimate is cruder for large values of $\beta$ than (\ref{3.9}) since it is linear in $\beta$,
it complements (\ref{3.9}) when $\beta$ assumes smaller values.

\medskip

In conclusion, we have shown that the minimization problem proposed in the study of Bolognesi and Tong \cite{BT} for the determination of an AdS monopole wall has a unique solution
which may be obtained by finding a correct initial slope of the solution to the  initial value problem associated with the  differential equation governing the Higgs profile function. Furthermore, some estimates for
the minimum energy are also obtained which may be used to calculate the position of the monopole wall explicitly in an asymptotic large-parameter limit, involving the Higgs mass, Higgs vacuum level, and the AdS
radius of curvature.

\end{document}